\documentclass[a4paper,10pt]{article}
\usepackage{amssymb}
\usepackage[leqno]{amsmath}
\usepackage{mathtools}
\usepackage{booktabs}
\usepackage{tcolorbox}
\usepackage{thmtools}
\usepackage{subcaption,thm-restate}
\usepackage[mathcal,mathscr]{eucal}
\usepackage{epsfig,graphicx,graphics,color}
\usepackage{caption}
\usepackage{enumerate}
\usepackage{algorithm}
\usepackage{algpseudocode}
\usepackage{enumitem}
\usepackage[sort]{cite}
\usepackage{titling}
\usepackage{hyperref}
\usepackage{xspace}
\usepackage{tikz}
\usetikzlibrary{arrows,backgrounds,positioning,calc,fit,shapes.geometric,decorations.pathreplacing,decorations.markings, tikzmark,topaths}
\usepackage{tkz-euclide}

\tikzstyle{internal} = [draw, fill, shape=circle, text=black, inner sep=4pt]
\tikzstyle{external} = [shape=circle, draw]
\tikzstyle{square}   = [draw, fill, rectangle, inner sep=5pt]
\tikzset{lnode/.style = {
		circle, 
		draw=cyan!30!black, 
		thick,
		inner sep=1.5pt,
		minimum size=15pt } }
\tikzset{uedge/.style = {
		draw=cyan!20!black, 
		very thick} }
\hypersetup{
	colorlinks=true,       % false: boxed links; true: colored links
	linkcolor=blue,        % color of internal links (change box color with linkbordercolor)
	citecolor=red,         % color of links to bibliography
	filecolor=magenta,     % color of file links
	urlcolor=cyan,         % color of external links
	linktocpage=true
}

\renewcommand{\title}[1]{\vspace{\fill}
	\eject\addtolength{\baselineskip}{4pt}
	{\bfseries\LARGE #1}\\[3mm]\addtolength{\baselineskip}{-4pt}}
\renewcommand{\author}[3]{\parbox[t]{75mm}
	{\begin{center}{\scshape #1}\\[3mm] #2\\
			{\ttfamily #3} \end{center}}}

\newtheorem{thm}{\bfseries Theorem}
\newtheorem{lem}[thm]{\bfseries Lemma}        %% lemmas, props, cor, etc
\newtheorem{remark}[thm]{\bfseries Remark}    %%   are numbered consecutively
\newtheorem{prop}[thm]{\bfseries Proposition} %%
\newtheorem{cor}[thm]{\bfseries Corollary}

\DeclareMathOperator\opt{\textsc{OPT}}
\DeclareMathOperator\alg{\textsc{ALG}}

\newenvironment{proof}{\medskip                    %% Proof
	\noindent{\scshape Proof:}}{\quad $\Box$\medskip}  %%
\usepackage{sectsty}
\usepackage{titlesec}
\titlespacing\section{0pt}{10pt}{2pt}
\titlespacing\subsection{0pt}{8pt}{2pt}
\titlespacing\subsubsection{0pt}{6pt}{2pt}
\makeatletter
\renewcommand{\paragraph}{%
	\@startsection{paragraph}{4}%
	{\z@}{1.3ex \@plus 1ex \@minus 0.2ex}{-1em}%
	{\normalfont\normalsize\bfseries}%
}
\makeatother
\makeatletter
\def\thm@space@setup{\thm@preskip=0pt
	\thm@postskip=4pt}
\makeatother

\def\final{1}  % set this to 1 to get a comment-free version
\def\iflong{\iffalse}
\ifnum\final=0  %namely if we allow comments in the output
\newcommand{\kanote}[1]{{\color{blue}[{\tiny \textbf{Karthik:} \bf #1}]\marginpar{\color{blue}*}}}
\newcommand{\knote}[1]{{\color{red}[{\tiny \textbf{Kristóf:} \bf #1}]\marginpar{\color{red}*}}}
\newcommand{\tnote}[1]{{\color{blue}[{\tiny \textbf{Tamás:} \bf #1}]\marginpar{\color{blue}*}}}
\newcommand{\dnote}[1]{{\color{orange}[{\tiny \textbf{Dani:} \bf #1}]\marginpar{\color{orange}*}}}
\else % in this case [final=1] we don't want any comments to show
\newcommand{\kanote}[1]{}
\newcommand{\knote}[1]{}
\newcommand{\tnote}[1]{}
\newcommand{\dnote}[1]{}
\fi

\DeclareMathOperator*{\argmin}{arg\,min}

\newcommand{\bR}{\mathbb{R}}

\newcommand{\bZ}{\mathbb{Z}}

\newcommand{\cI}{\mathcal{I}}

\newcommand{\cB}{\mathcal{B}}

\newcommand{\cH}{\mathcal{H}}

\newcommand{\cP}{\mathcal{P}}
\newcommand{\cM}{\mathcal{M}}

\newcommand{\subkpart}{\textsc{Sub-k-P}\xspace}
\newcommand{\symsubkpart}{\textsc{SymSub-k-P}\xspace}
\newcommand{\monsubkpart}{\textsc{MonSub-k-P}\xspace}
\newcommand{\submpart}{\textsc{Sub-MP}\xspace}
\newcommand{\symsubmpart}{\textsc{SymSub-MP}\xspace}
\newcommand{\monsubmpart}{\textsc{MonSub-MP}\xspace}
\newcommand{\gcovmpart}{\textsc{GCov-MP}\xspace}
\newcommand{\mcut}{\textsc{Multiway-Cut}\xspace}
\newcommand{\kcut}{\textsc{$k$-Cut}\xspace}
\newcommand{\submcp}{\textsc{Sub-MCP}\xspace}
\newcommand{\symsubmcp}{\textsc{SymSub-MCP}\xspace}
\newcommand{\monsubmcp}{\textsc{MonSub-MCP}\xspace}
\newcommand{\mcmcut}{\textsc{MC-Multiway-Cut}\xspace}
\newcommand{\mcmcutdoub}{\textsc{Double MC-Multiway-Cut}\xspace}
\newcommand{\mcmcutcom}{\textsc{Common MC-Multiway-Cut}\xspace}
\newcommand{\gcovmcp}{\textsc{GCov-MCP}\xspace}

\newcommand{\symsubdmcp}{\textsc{SymSub-DMCP}\xspace}

\newtcolorbox{probbox}{arc=6pt,
	colback=white!100,
	colframe=black!50,
	before skip=6pt,
	after skip=6pt,
	boxsep=1pt,
	left=6pt,
	right=6pt,
	top=4pt,
	bottom=4pt}

\newcommand{\searchprob}[3]{{\small
		\begin{center}%  
			\begin{minipage}{\linewidth}%
				\begin{probbox}
					\textsc{#1}\\[0.4ex]
					\textbf{Input:} #2\\[0.4ex]
					\textbf{Goal:} #3
				\end{probbox}
			\end{minipage}%
		\end{center}
}}

\begin{document}
	\begin{center}
		
		%%%%%%%%%%%%%%%%%%%%%%%%%%%%%%%%%%%%%%%%%%%%%%%%%%%%%%%%
		% Title
		%%%%%%%%%%%%%%%%%%%%%%%%%%%%%%%%%%%%%%%%%%%%%%%%%%%%%%%%
		\title{Approximating Submodular Matroid-Constrained Partitioning}
		%%%%%%%%%%%%%%%%%%%%%%%%%%%%%%%%%%%%%%%%%%%%%%%%%%%%%%%%
		% begin : Authors
		%%%%%%%%%%%%%%%%%%%%%%%%%%%%%%%%%%%%%%%%%%%%%%%%%%%%%%%%
		\author{Kristóf Bérczi\footnotemark[1]  
		}{
			MTA-ELTE Matroid Optimization Group \\ HUN-REN--ELTE Egerv\'{a}ry Research Group\\
			Department of Operations Research \\
			E\"{o}tv\"{o}s Lor\'{a}nd
			University, Budapest, Hungary
		}{
			kristof.berczi@ttk.elte.hu
		}\footnotetext[1]{The research was supported by the Lend\"ulet Programme of the Hungarian Academy of Sciences -- grant number LP2021-1/2021, by the Ministry of Innovation and Technology of Hungary from the National Research, Development and Innovation Fund -- grant numbers ADVANCED 150556 and ELTE TKP 2021-NKTA-62, and by Dynasnet European Research Council Synergy project -- grant number ERC-2018-SYG 810115.}
		%%%%%%%%%%%%%%%%%%%%%%%%%%%%%%%%%%%%%%%%%%%%%%%%%%%%%
		% Add the text of ``thanks'' above
		%%%%%%%%%%%%%%%%%%%%%%%%%%%%%%%%%%%%%%%%%%%%%%%%%%%%%
		\author{
			Karthekeyan Chandrasekaran\footnotemark[2] 
			%%%%%%%%%%%%%%%%%%%%%%%%%%%%%%%%%%%%%%%%%%%%
			%% Underline the name of the speaker
			%%%%%%%%%%%%%%%%%%%%%%%%%%%%%%%%%%%%%%%%%%%
		}{
			Grainger College of Engineering, 
			\\
			University of Illinois,
			\\
			Urbana-Champaign, USA
		}{
			karthe@illinois.edu
		}
		\footnotetext[2]{Supported in part by NSF grant CCF-2402667.}
		\author{
			Tamás Király\footnotemark[1]
		}{ 
			HUN-REN{--}ELTE Egerváry Research Group\\
			Department of Operations Research \\
			Eötvös Loránd University, Budapest, Hungary
		}{
			tamas.kiraly@ttk.elte.hu
		}
		\author{
			Daniel P. Szabo\footnotemark[1]
		}{ 
			Department of Operations Research \\
			Eötvös Loránd University, Budapest, Hungary
		}{
			dszabo2@wisc.edu
		}
		
		%%%%%%%%%%%%%%%%%%%%%%%%%%%%%%%%%%%%%%%%%%%%%%%%%%%%%%%%
		% end : Authors
		%%%%%%%%%%%%%%%%%%%%%%%%%%%%%%%%%%%%%%%%%%%%%%%%%%%%%%%%
		
	\end{center}
	
	%%%%%%%%%%%%%%%%%%%%%%%%%%%%%%%%%%%%%%%%%%%%%%%%%%%%%%%%
	% Abstract
	%%%%%%%%%%%%%%%%%%%%%%%%%%%%%%%%%%%%%%%%%%%%%%%%%%%%%%%%
	
	\begin{quote}
		{\bfseries Abstract:}
		The submodular partitioning problem asks to minimize, over all partitions $\mathcal{P}$ of a ground set $V$, the sum of a given submodular function $f$ over the parts of $\mathcal{P}$. The problem has seen considerable work in approximability, as it encompasses multiterminal cuts on graphs, $k$-cuts on hypergraphs, and elementary linear algebra problems such as matrix multiway partitioning. This research has been divided between the fixed terminal setting, where we are given a set of terminals that must be separated by $\mathcal{P}$, and the global setting, where the only constraint is the size of the partition. We investigate a generalization that unifies these two settings: minimum submodular matroid-constrained partition. In this problem, we are additionally given a matroid over the ground set and seek to find a partition $\mathcal{P}$ in which \emph{there exists} some basis that is separated by $\mathcal{P}$. We explore the approximability of this problem and its variants, reaching the state of the art for the special case of symmetric submodular functions, and provide results for monotone and general submodular functions as well.
	\end{quote}
	
	%%%%%%%%%%%%%%%%%%%%%%%%%%%%%%%%%%%%%%%%%%%%%%%%%%%%%%%%
	% Keywords (3 $\sim$ 5 words)
	%%%%%%%%%%%%%%%%%%%%%%%%%%%%%%%%%%%%%%%%%%%%%%%%%%%%%%%%
	\begin{quote}
		{\bf Keywords:} Approximation algorithms, Matroid constraints, Multiway cut, Submodular partitioning
	\end{quote}
	\vspace{5mm}

	%%%%%%%%%%%%%%%%%%%%%%%%%%%%%%%%%%%%%%%%%%%%%%%%%%%%%%%%
	% Text
	%%%%%%%%%%%%%%%%%%%%%%%%%%%%%%%%%%%%%%%%%%%%%%%%%%%%%%%%
	
	%%%%%%%%%%%%%%%%
	\section{Introduction}
	%%%%%%%%%%%%%%%%
	
	Many motivations, as well as techniques, for submodular optimization stem from graph partitioning problems, making them a natural starting point. Given an undirected graph $G=(V,E)$, edge weights $w\in\bR^E_+,$ and a set $T\subseteq V$ of $k$ terminals, the \emph{Multiway Cut} (\mcut) problem asks for a partition $\cP=(V_1,\dots, V_k)$ of $V$ minimizing $\sum_{i=1}^k d_w(V_i)$ such that each part $V_i$ contains exactly one terminal from $T$. A well-known special case is $k=2$, which corresponds to the polynomial-time solvable \emph{Min-cut} problem. For $k \geq 3$, the problem is NP-hard~\cite{Dahlhausetal} but admits approximation algorithms~\cite{SharmaVondrak}. In contrast, the \emph{$k$-cut} problem (\kcut) takes as input only the graph $G = (V, E)$, edge weights $w\in\mathbb{R}^E_+$, and an integer $k$, and asks for a partition $\cP=(V_1,\dots, V_k)$ of $V$ into $k$ non-empty parts minimizing  $\sum_{i=1}^k d_w(V_i)$. This problem is polynomial-time solvable for every fixed $k$, but NP-hard when $k$ is part of the input. In this case, a $(2-2/k)$-approximation is possible~\cite{kcut2apx}, and this approximation is tight under the Small Set Expansion Hypothesis~\cite{kcut}.
	
	These problems have been generalized to arbitrary submodular functions. A set function $f\colon 2^V\to \bR$ is \emph{submodular} if $f(A) + f(B)\geq f(A\cup B) + f(A\cap B) $ for all $A,B\subseteq V$. Additionally, it is \emph{symmetric} if $f(A)=f(V \setminus A)$ for all $A\subseteq V$, and \emph{monotone} if $f(A)\leq f(B)$ for all $A\subseteq B\subseteq V$. For ease of discussion, we assume throughout that $f(\emptyset) = 0$. In the \emph{Submodular $k$-partition Problem} (\subkpart), we are given a submodular function $f\colon 2^V\to \bR_{\ge 0}$ and an integer $k$, and the goal is to find a partition $\cP=\{V_1,\dots, V_k\}$ of $V$ into $k$ non-empty parts minimizing $\sum_{i=1}^k f(V_i)$. When $f$ is the cut function of an undirected graph, then it is symmetric and submodular, hence we get back the \kcut problem. \subkpart does not admit a constant-factor approximation assuming the Exponential Time Hypothesis (ETH)~\cite{k-wayhypergraphcuthardness}, and the best known approximation factor is $O(k)$~\cite{greedysplit, submod_kapx}. However, for certain special classes of submodular functions, it admits constant factor approximations. For symmetric submodular functions (\symsubkpart), the problem admits a tight $(2-2/k)$-approximation~\cite{greedysplit, ppskarthikwang, santiago-for-lower-bound}, while for monotone submodular functions (\monsubkpart), the problem admits a tight $4/3$-approximation ~\cite{ppskarthikwang, santiago-for-lower-bound}. 
	
	In the \emph{Submodular Multiway Partition} problem (\submpart), we are given a submodular function $f\colon 2^V\to \bR_+$ and $k$ terminals $t_1, \dots, t_k$, and the goal is to find a partition $\cP=(V_1,\dots, V_k)$ with $t_i\in V_i$ for every $i\in [k]$ minimizing $\sum_{i=1}^k f(V_i)$. When $f$ is the cut function of an undirected graph, then we get back the \mcut problem. \submpart admits a $2$-approximation~\cite{chekuri-ene}. For symmetric submodular functions (\symsubmpart), the problem admits a $3/2$-approximation~\cite{chekuri-ene}, while for monotone submodular functions (\monsubmpart), the problem admits a $4/3$-approximation. An important special case of \monsubmpart is when $f$ is a \emph{graph coverage function} (\gcovmpart), that is, $f(S)\coloneqq \sum_{uv\in E \colon  \{u,v\}\cap S \neq \emptyset} w(uv)$ for every $S\subseteq V$. We note that the objective function of \gcovmpart differs from the objective function of \mcut  by an additive factor of $|E|$. Yet, \gcovmpart differs substantially from \mcut  in terms of approximability: \gcovmpart admits a $1.125$-approximation~\cite{fixedterms}, while the best known approximation for \mcut is $1.2965$~\cite{SharmaVondrak}.
	
	Algorithmic techniques for global and fixed terminal partitioning problems differ significantly: the key difference is that fixed terminal versions allow for linear programming techniques. On the other hand, approximation algorithms for global partitioning problems have relied on global structural aspects of submodular functions, such as Gomory-Hu trees, greedy splitting, and cheapest singleton algorithms. Nevertheless, the approximation results are not always so different. For example, both \monsubkpart and \monsubmpart admit a $4/3$-approximation. Previous work has also generalized these two regimes in some way, such as in \cite{prev_ppr} and \cite{steiner-k-cut}. One of the motivations of this work is to present a general model for partitioning problems that unifies the global and fixed-terminal variants, and to investigate approximation algorithms for this general model. The choice of encoding the terminal choices as the bases of some matroid is not without precedent: the same was effective for the minimum Steiner tree problem~\cite{matroidsteiner}.
	
	%%%%%%%%%%%%%%%%
	\subsection{Results}
	%%%%%%%%%%%%%%%%
	
	We introduce a new problem that unifies the global and the fixed terminal partitioning problems. The \emph{Submodular Matroid-constrained Partitioning} problem (\submcp) is defined as follows:
	
	\searchprob{\submcp}{A submodular function $f\colon 2^V\to \bR_+$ given by a value oracle, and a rank-$k$ matroid $\cM=(V, \cB)$ given by an independence oracle, where $\cB$ is the family of bases of $\cM$.}{Find a solution to
		\[ \min\left\{ \sum_{i=1}^k f(V_i) \colon  \{ V_i\}_{i=1}^k \text{ is a partition of } V \text{ s.t. there is }
		B\in \cB \text{ with } |B\cap V_i|=1 \text{ for all } i\in [k] \right\}.\]}
	
	In other words, our aim is to find a $k$-partition that separates the elements of some basis of the matroid. If the input function is symmetric or monotone submodular, then we denote the resulting problem as \symsubmcp or \monsubmcp, respectively. A special case of \symsubmcp is \emph{Matroid-constrained Multiway Cut} (\mcmcut), formally defined as follows. 
	
	\searchprob{\mcmcut}{A graph $G=(V,E)$, edge weights $w\in\bR^E_+$, and a rank-$k$ matroid $\cM=(V, \cB)$ given by an independence oracle, where $\cB$ is the family of bases of $\cM$.}{Find a solution to
		\[ \min\left\{ \sum_{i=1}^k d_w(V_i) \colon \{ V_i\}_{i=1}^k \text{ is a partition of } V \text{ s.t. there is }
		B\in \cB \text{ with } |B\cap V_i|=1 \text{ for } i\in [k] \right\}.\]}
	
	We study several algorithmic approaches for various special classes of matroids. \submcp contains \emph{Hypergraph $k$-cut} as a special case, so it does not admit a constant-factor approximation assuming the ETH. The first problem we study is \symsubmcp. In this case, many of the algorithms for \subkpart extend to \submcp while achieving the same approximation guarantees.
	
	Algorithm~\ref{alg:gh_greedy} constructs a Gomory-Hu tree for a given symmetric submodular function and then greedily selects the minimum-weight edges for which a basis exists that respects the partition induced by the chosen edges. 
	\begin{thm}\label{thm:symGH}
		The Gomory-Hu tree algorithm provides a $(2-2/k)$-approximation for \symsubmcp.
	\end{thm}
	
	%The Gomory-Hu tree approach also extends to symmetric submodular partitioning under a certain type of matroid intersection constraint -- we describe the details in Section \ref{sec:intersection}.
	
	The greedy splitting algorithm from~\cite{greedysplit} proceeds in a greedy fashion somewhat differently than the Gomory-Hu tree approach. It first finds the cheapest submodular cut $S$ of $V$, and then iteratively finds, for a given intermediate partition $\cP_i$, the next cheapest $\argmin \{f(X) +f(W \setminus X) - f(W):X\subset W, W\in \cP_{i}\}$ split to increase the size of the partition by one. This algorithm, as we show in Lemma \ref{lem:main}, generalizes to the matroid-constrained versions as well to get the same $(2-2/k)$-approximation factor.
	
	\begin{thm}\label{thm:greedysplitsym}
		The greedy splitting algorithm provides a $(2-2/k)$-approximation for \symsubmcp.
	\end{thm}
	
	Next, we turn to monotone submodular $k$-partitioning. Here, a natural adaptation of the cheapest singleton algorithm is a fast $(2-1/k)$-approximation. In the problem without the matroid constraint, the cheapest singleton algorithm proceeds by choosing $k-1$ elements with the cheapest singleton cost, namely such that $f(v_1)\leq  \dots \leq f(v_{k-1}) \leq f(v_j)$ for all $j\geq k$. It then returns the partition $(\{v_1\},\dots, \{v_{k-1}\}, \{v_k,\dots, v_n\})$. To see why this yields a $(2-1/k)$-approximation, let $\cP^* = \{V_1^*,\dots, V_k^*\}$ be the optimal partition, ordered such that $f(V_1^*)\leq\dots\leq f(V_k^*)$. By monotonicity, we have $\sum_{i=1}^{k-1} f(v_i) \leq \sum_{i=1}^{k-1} f(V_i^*) \leq (1-1/k) \sum_{i=1}^k f(V_i^*) $. Moreover, monotonicity and submodularity imply that $f(\{v_k,\dots, v_n\}) \leq f(V) \leq \sum_{i=1}^k f(V_i^*)$. For the matroid-constrained setting, we can weight the matroid by the value of $f$ on the singletons. We then select a minimum weight
	independent set of size $k-1$ in the matroid, and return the partition where these $k-1$ elements are singleton classes, and the last class consists of the remaining elements. The approximation factor follows by the same argument. 
	
	We can obtain a slightly better bound using the greedy splitting technique from \cite{greedysplit} by extending their main lemma to the matroid-constrained setting.
	
	\begin{thm}\label{thm:greedysplitmon}
		The greedy splitting algorithm provides a $(2-2/k)$-approximation for \monsubmcp.
	\end{thm}
	
	We additionally show that this analysis is tight for the greedy splitting algorithm, even without the matroid constraints, i.e., there exist instances of \monsubkpart for which the approximation factor of greedy splitting is at least $(2-2/k)$ -- see Lemma~\ref{lem:tightness}. 
	
	Another consequence of the applicability of the greedy splitting algorithm is that we can extend the linear approximation guarantee from \cite{greedysplit} to general \submcp. We state this in Theorem \ref{thm:general}.
	
	\begin{thm}\label{thm:general}
		The greedy splitting algorithm provides a $(k-1)$-approximation for \submcp.
	\end{thm}
	
	Finally, we consider the special case of graph coverage functions. We denote by \gcovmcp the special case of \submcp where the input function $f$ is a graph coverage function, defined as $f(S)\coloneqq \sum_{uv\in E \colon  \{u,v\}\cap S \neq \emptyset} w(uv)$ for every $S\subseteq V$. Note that this function is monotone, and only differs from graph cut functions in terms of approximation, as $f(\cP) = d_w(\cP) + |E|$ for any partition $\cP$, where $d_w(\cP)$ denotes the total weight of edges whose two endpoints are in diferent classes of $\cP$.
	In this case, we can improve upon the result for monotone functions to achieve a factor that is identical to the current best factor for the global case. The proof is short and simple, so we include it here.
	
	\begin{thm}\label{thm:gcov}
		%The partition output by Algorithm \ref{alg:gh_greedy}, 
		The Gomory-Hu tree algorithm provides a $4/3$-approximation for \gcovmcp.
	\end{thm}
	
	\begin{proof}
		Let $\cP$ be the partition returned by the Gomory-Hu tree algorithm 
		%Algorithm \ref{alg:gh_greedy}, 
		and $\cP^*$ be an optimal partition minimizing the graph coverage function $f$, which is also the partition minimizing $d_w$. Then, Theorem \ref{thm:symGH} shows $d_w(\cP) \leq (2-2/k) d_w(\cP^*) \leq 2 d_w(\cP^*)$. Let $c = \frac{d_w(\cP^*)}{|E|}$. Then
		\begin{align*}
			\frac{f(\cP)}{f(\cP^*)} \leq  \frac{2d_w(\cP^*)+|E|}{d_w(\cP^*)+|E|} =  \frac{2 d_w(\cP^*)+ d_w(\cP^*)/c}{d_w(\cP^*)+d_w(\cP^*)/c} = \frac{2c+1}{c+1}.
		\end{align*}
		On the other hand, $d_w(\cP^*)+|E| = (c+1)|E|$, and trivially $f(\cP) \leq 2|E|,$ so
		\begin{align*}
			\frac{f(\cP)}{f(\cP^*)} \leq  \frac{2|E|}{(c+1)|E|} =  \frac{2}{c+1}.
		\end{align*}
		Thus $\frac{f(\cP)}{f(\cP^*)} \leq \max_{0\leq c\leq 1}\min\left\{ \frac{2c+1}{c+1}, \frac{2}{c+1} \right\} = 4/3$, with the maximum attained at $c=1/2$.
	\end{proof}
	
	Our results are summarized in Table~\ref{tab:super_tab}.
	
	%%%%%%%%%%%%%%%%%%%%%%%%%%%%%%%%%%%%%%%%%%%%%%%%%%%%%
	% Table!
	\begin{table}[!h]
		\centering
		\begin{tabular}{|c|c|c|c|}
			\hline
			% \multicolumn{2}{|c|}{Matroid Constrained} \\
			&Matroid-Constrained & Global & Fixed Terminals  \\
			\hline
			General& $O(k)$-apx (Thm. \ref{thm:general}) & $O(k)$-apx \cite{submod_kapx, greedysplit} & $2$-apx \cite{chekuri-ene} \\ 
			\hline
			Symmetric& $(2-2/k)$-apx (Thms. \ref{thm:symGH}, \ref{thm:greedysplitsym}) & $(2-2/k)$-apx \cite{greedysplit} & $(3/2)$-apx \cite{chekuri-ene} \\ 
			\hline
			Monotone& $(2-2/k)$-apx (Thm. \ref{thm:greedysplitmon}) & $(4/3)$-apx \cite{ppskarthikwang} & $(4/3)$-apx \cite{fixedterms} \\ 
			\hline
			Graph Coverage& $(4/3)$-apx (Thm. \ref{thm:gcov}) & $(4/3)$-apx \cite{ppskarthikwang} & $(9/8)$-apx \cite{fixedterms} \\ 
			\hline
		\end{tabular}
		\caption{Comparison of our approximation results (first column) with the current best for global and fixed-terminal settings across different classes of submodular functions.}
		\label{tab:super_tab}
	\end{table}
	%%%%%%%%%%%%%%%%%%%%%%%%%%%%%%%%%%%%%%%%%%%%%%%%%%%%%
	
	%%%%%%%%%%%%%%%%
	\section{Preliminaries}
	%%%%%%%%%%%%%%%%
	
	\paragraph{Notation.}
	
	We denote by $\bR_+$ and $\bZ_+$ the sets of nonnegative reals and integers, respectively. For a positive integer $k$, we use $[k]\coloneqq \{1,\dots,k\}$. For a set function $f:2^V\to \bR_+$ and a partition $\cP$ of $V$, we use the notation $f(\cP) \coloneqq \sum_{X\in \cP}f(X)$. For a graph $G=(V,E)$ and $X\subseteq V$, let $\delta_G(X)$ denote the set of edges leaving $X$, that is, $\delta_G(X) \coloneqq \{uv\in E\colon |\{u,v\}\cap X|=1\}$. The cut function of $G$ is then defined as $d_G(X) \coloneqq |\delta_G(X)|$ for $X\subseteq V$. We omit the subscript $G$ when the graph is clear from the context. If edge weights $w\in\bR^E_+$ are given, we use $d_w(X) \coloneqq \sum_{e\in \delta(X)} w(e)$. In directed graphs, we use $\delta^{in}, d^{in}$, and $\delta^{out}, d^{out}$ for incoming and outgoing edges, respectively. For a graph $G=(V,E)$ and a partition $\cP=\{V_1,\dots,V_k\}$ of $V$, $\delta_G(\cP) \coloneqq \cup_{i=1}^k \delta_G(V_i)$ is called the \emph{boundary} of $\cP$. If edge weights $w\in\bR^E_+$ are also given, we use $d_w(\cP)=w(\delta_G(\cP))$.
	
	For a subset $F\subseteq E$, we denote the graph obtained by deleting the edges in $F$ by $G-F$. When inserting (deleting) a single element $e$ from a set $X$, we often use $X+e$ $(X-e)$ for $X\cup \{e\}$ $(X\setminus \{e\})$.
	
	\paragraph{Matroids.} A \emph{matroid} $\cM=(V,\cI)$ is defined by its \emph{ground set} $V$ and its \emph{family of independent sets} $\cI\subseteq 2^V$ that satisfies the \emph{independence axioms}: (I1) $\emptyset\in\cI$, (I2) $X\subseteq Y,\ Y\in\cI\Rightarrow X\in\cI$, and (I3) $X,Y\in\cI,\ |X|<|Y|\Rightarrow\exists e\in Y-X\ s.t.\ X+e\in\cI$. The subsets of $S$ not in $\cI$ are called \emph{dependent}. The \emph{rank} $r(X)$ of a set $X$ is the maximum size of an independent set in $X$. The \emph{rank} of the matroid is $r(V)$. The maximal independent subsets of $V$ are called \emph{bases} and their family is usually denoted by $\cB$. The matroid is uniquely defined by its rank function or by its family of bases, so we also use the notation $\cM=(V,r)$ or $\cM=(V,\cB)$.  An inclusionwise minimal non-independent set is a \emph{circuit}. 
	
	Given a matroid $\cM=(V,\cI)$, its \emph{$k$-truncation} is the matroid $(V,\cI_k)$ with $\cI_k=\{X\in\cI\mid |X|\leq k\}$. We denote the $k$-truncation of $\cM$ by $(\cM)_{k}$. For a set $Z\subseteq V$, the \emph{contraction of $Z$} results in a matroid $\cM/Z=(V \setminus Z,\cI')$ in which a set $X\subseteq V \setminus Z$ is independent if and only if $X\cup I$ is independent for any maximum independent set of $\cM$ in $Z$. Let $D=(V,A)$ be a directed graph, $V_1,V_2\subseteq V$, and $M_1=(V_1,\cI_1)$ be a matroid. We say that a set $Y\subseteq V_2$ is \emph{linked} to another set $X\subseteq V_1$ if $|X|=|Y|$ and there exists $|X|$ pairwise vertex-disjoint directed paths from $X$ to $Y$. Then the matroid $M_2=(V_2,\cI_2)$ defined by $\cI_2\coloneqq\{I_2\subseteq V_2\colon \text{$I_2$ is linked to $I_1$ for some $I_1\in\cI_1$} \}$ is called the \emph{matroid of $M_1$ induced by $D$ on $V_2$}.
	
	A \emph{laminar matroid} is a matroid $\cM=(V,\cI)$ where $\cI=\{F\subseteq V\colon |F\cap V_i|\leq g_i\ \text{for each $i\in [q]$}\}$ for some laminar family $V_1,\dots, V_q$ of subsets of $V$ and upper bounds $g_1,\dots,g_q\in\bZ_+$. A rank-$r$ matroid $\cM=(V,\cI)$ is \emph{paving} if every set of size at most $r-1$ is independent, or equivalently, if every circuit has size at least $r$. Such matroids have a nice characterization using hypergraphs~\cite{hartmanis1959lattice,welsh1976matroid}.
	
	\begin{prop}
		For a non-negative integer $r$, a ground set $V$ of size at least $r$, and a (possibly empty) family $\mathcal{H}=\{H_1,\dots,H_q\}$ of proper subsets of $V$ such that $|H_i\cap H_j|\leq r-2$ for $1\leq i < j\leq q$, the family $\mathcal{B}_{\mathcal{H}}=\{X\subseteq V\mid |X|=r,\ X\not\subseteq H_i\ \text{for } i=1,\dots,q\}$ forms the set of bases of a paving matroid. Furthermore, every paving matroid can be obtained in this way.
	\end{prop}
	
	As is common in matroid algorithms, we assume that matroids are given by an independence oracle, and we measure the running time in terms of the number of oracle calls and other elementary operations. For simplicity, we refer to an algorithm as ``polynomial time'' if its number of oracle calls and elementary operations is polynomial in the size of the ground set.

	%%%%%%%%%%%%%%%%
	\section{Symmetric Submodular Functions}
	%%%%%%%%%%%%%%%%
	
	We prove Theorem~\ref{thm:symGH} in this section. The proof builds on a method outlined in \cite{prev_ppr}. First, we show that \mcmcut is solvable on trees by a type of greedy algorithm. Then we use this on the Gomory-Hu tree of a symmetric submodular function to obtain an approximation algorithm for \symsubmcp.
	
	%%%%%%%%%%%%%%%%
	\subsection{\texorpdfstring{MC-Multiway-Cut}{\mcmcut} on Trees}
	%%%%%%%%%%%%%%%%
	
	As a first step, we show that \mcmcut is solvable in polynomial time on trees. Interestingly, in this case the problem reduces to finding a minimum weight basis of a matroid. Recall that in this setting, we are given a tree $G=(V,E)$ together with edge weights $w\in\bR^E_+$ and a matroid $\cM=(V,\cI)$. Let us define the following family of subsets of $E$:
	\begin{align} \label{eq:good}
		\cI' = \{ X'\subseteq E \colon \text{the components $V_1,\dots,V_{|X'|+1}$ of $G-X'$ have a transversal independent in $\cM$}\}.
	\end{align}
	In other words, an edge set $X'$ is in $\cI'$ if there exists $v_i\in V_i$ for every connected component $V_i$ of $G-X'$ such that $\{v_1,\dots,v_{|X'|+1}\}$ forms an independent set of $\cM$.
	
	\begin{lem}\label{lem:tree_matroid}
		If $\cM$ has rank at least $1$, then $\cI'$ satisfies the independence axioms.
	\end{lem}
	\begin{proof}
		We verify the axioms one by one. Axiom (I1) clearly holds, since the rank of $\cM$ is at least one, and any non-loop $v\in V$ is a transversal.
		
		To see (I2), let $Y'\in\cI'$ and $X'\subseteq Y'$. Let $\mathcal{V}$ and $\mathcal{U}$ be the families of connected components of $G-Y'$ and $G-X'$, respectively. Since $X'\subseteq Y'$, the partition $\mathcal{U}$ forms a coarsening of $\mathcal{V}$. By the definition of $\cI'$, $Y' \in \cI'$ implies that there exists a transversal $Y$ of $\mathcal{V}$ that is independent in $\cM$. Since $\mathcal{U}$ is a coarsening of $\mathcal{V}$, a properly chosen subset $X$ of $Y$ will form a transversal of $\mathcal{U}$, thus showing that $X' \in \cI'$.
		
		Finally, we verify (I3). Take any $X', Y'\in \cI'$ with $|X'|<|Y'|$. Let $X,Y \in \cI$ be the corresponding independent sets in $\cM$ chosen such that $|X\cap Y|$ is as large as possible. By axiom (I3) for $\cI$, there exists $u\in Y\setminus X$ such that $X+u\in \cI$. Let $C$ be the component of $G-X'$ containing $u$, and $v$ be the element of $X$ in $C$. If $v\notin Y$, then $X-v+u\in\cI$ is still a transversal of the conencted somponents of $G-X'$, contradicting the choice of $X$ and $Y$. Thus $v\in Y$. As $u,v\in Y$, there exists an edge $e\in Y'$ that lies on the path between $u$ and $v$ in $G$. Then $X'+e$ is in $\cI'$ as desired, as it has $X+u$ as an independent set of representatives.
	\end{proof}
	
	Given the independence oracle access to $ \cM $, we can construct an independence oracle to the matroid $ \cM'=(E,\cI') $, as deciding whether there is an independent set of representatives for a given partition of $ V $ reduces to the intersection of a partition matroid and $ \cM $.

	\begin{remark}
		It is worth mentioning that $\cI'$ can be identified with the independent set of a matroid of $\cM$ induced by a directed graph. To see this, pick an arbitrary vertex $r\in V$, and orient the edges of $G$ such that the resulting directed graph is an in-arborescence rooted at $r$. Then, subdivide every arc $uv$ with a new vertex $x_{uv}$, resulting in arcs $ux_{uv}$ and $x_{uv}v$. Finally, add a new vertex $x$ to the graph together with the arc $rx$. Note that the vertex set of the directed graph thus obtained is $V\cup\{x_{uv}\colon uv\in E\}\cup\{x\}$. Now consider the matroid obtained by taking the matroid of $\cM$ induced by the directed graph on $\{x_{uv}\colon uv\in E\}\cup\{x\}$, and contracting $x$ in it. It is not difficult to check that $Y\subseteq \{x_{uv}\colon uv\in E\}$ is independent in this matroid if and only if the corresponding subset of edges is in $\cI'$.
	\end{remark}
	
	\begin{cor}\label{lem:trees}
		\mcmcut is solvable in polynomial time on trees.
	\end{cor}
	\begin{proof}
		First, we observe that \mcmcut has an optimal solution $\cP=\{V_1,\dots, V_k\}$ where $G[V_i]$ is connected for every $i\in [k]$. To show this, we choose $\cP$ so that $|\delta_G(\cP)|$ is smallest, and choose a basis $B$ with $|B \cap V_i|=1$ for all $i \in [k]$. Suppose for contradiction that some $G[V_i]$ is not connected, and let $U$ be a component that does not contain $B\cap V_i$. Since $G$ is connected, there exists an edge from $U$ to some $V_j$. The partition $\cP'=\cP\setminus \{V_i,V_j\}\cup \{V_j\cup U, V_i\setminus U\}$ is also optimal and $|\delta_G(\cP')| < |\delta_G(\cP)|$, a contradiction.
		
		By the above observation and Lemma~\ref{lem:tree_matroid}, the minimum weight matroid-constrained multiway cut on trees corresponds to the minimum weight basis of the matroid $ \cM'=(E,\cI') $. Thus, \mcmcut can be solved in polynomial time by using the greedy algorithm to find a minimum weight basis of $\cM'$.
	\end{proof}
	
	%%%%%%%%%%%%%%%%
	\subsection{Symmetric Submodular Functions}
	%%%%%%%%%%%%%%%%
	\kanote{Title is confusing. Section is about Symmetric Submodular Functions and not General Submodular Functions.} \dnote{Good point, fixed}

	Goemans and Ramakrishnan~\cite{ghtree1} were the first to note that symmetric submodular systems have a Gomory-Hu tree.
	Given a symmetric submodular function $f$ over a ground set $V$, its Gomory-Hu tree is a tree $H=(V,F)$ with a weight function $w_H\in\bR^F_+$ that encodes the minimum $s-t$ cuts for each pair $s,t$ of vertices in the following sense: given a path $P_{s,t}$ between $s$ and $t$, $\min_{e\in P_{s,t}}w_H(e) = \min_{s\in S\subseteq V- t} f(S)$.
	Furthermore, the two components of the tree obtained by removing the edge of minimum $w_H$-weight on the path give the two sides of a minimum $s-t$ cut in the submodular system $f$.
	
	The greedy algorithm solves \mcmcut in the special case when $ G $ is a tree. The classical $ (2-2/k) $-approximation for \mcut uses 2-way cuts coming from the Gomory-Hu tree, and so does the $ (2-2/k) $-approximation for \kcut. We follow a similar approach to obtain a $ (2-2/k) $-approximation for \symsubmcp, presented in Algorithm~\ref{alg:gh_greedy}. The algorithm can be interpreted as taking the minimum edges in the Gomory-Hu tree as long as they allow a valid system of representatives.
	
	\begin{algorithm}[!h]
		\caption{Approximation algorithm for \symsubmcp}\label{alg:gh_greedy}
		\begin{algorithmic}[1]
			\Statex \textbf{Input:} A symmetric, submodular function $f\colon 2^V\to \bR_+$ and a rank-$k$ matroid $ \cM $.
			\Statex \textbf{Output:} A feasible partition $\cP$. 
			\State Compute the Gomory-Hu tree $ H=(V,E) $ of $ f $ and its weight function $w_H$.
			\State Let $\cM'=(E,\cI')$ be the matroid defined as in \eqref{eq:good} using $\cM$ and $H$.
			\State Set $C \gets \emptyset$.
			\While{$|C|< k-1$}
			\State $e\gets \argmin\{ w_H(e) \colon  e\notin C, C+e\in \cI' \}$
			\State $C \gets C+e$
			\EndWhile
			\State Return the connected components of $(V,E \setminus  C)$.
		\end{algorithmic}
	\end{algorithm}
	
	\vspace{-.1cm}
	\begin{thm}\label{thm:sym_greedy}
		Algorithm~\ref{alg:gh_greedy} provides a $ (2-2/k) $-approximation to \symsubmcp.
	\end{thm}
	\begin{proof}
		Let $ \opt =\{V_1^*,\dots,V_k^* \}$ be an optimal partition, and let $t_i^*\in V_i^*$ for $i\in[k]$ be representatives such that $\{t_1^*,\dots, t_k^*\} \in \cB $. Without loss of generality, we may assume that $ V_k^* $ has the maximum $f$ value among the partition classes. Let $ H=(V,E) $ be the Gomory-Hu tree of $ f $.
		
		We transform $\opt $ into a solution $\opt_{GH}$ on $H$, losing at most a factor of $(2-2/k)$.
		We do this by repeatedly removing the minimum weight edge in $E$ that separates a pair among the representatives $t_1^*,\dots, t_k^*$ that are in the same component of $H$. More precisely, we start with $H_0=H$, and take the minimum-weight edge $ e_1\in E(H_0) $ separating some pair of representatives $t_{i}^*, t_{j}^*$ that are in the same component of $ H_0 $. Define the edge $f_1=(t_{i}^*, t_{j}^*) $. Then we construct $ H_1=H_0-e_1 $, and repeat this process to get a sequence of edges $ e_1,\dots, e_{k-1} $ and a tree of representative pairs $ F=(\{ t_1^*,\dots, t_k^* \},\{f_1,\dots, f_{k-1}\}) $.
		
		Fix a vertex $x$ arbitrarily. For an edge $e\in E$, let $U(e)$ denote the vertices of the connected component of $H-e$ containing $x$. Direct the edges of $ F $ away from $ t_k^* $, and reorder the edges such that $ f^1 $ is the edge going into $ t_1^* $, $ f^2 $ into $ t_2^* $, and so on. Let $ e^i $ be the edge of the Gomory-Hu tree corresponding to $ f^i $, i.e., the minimum weight edge of the path between the two endpoints of $ f^i $.
		Then we have $ f(V_i^*) \geq f(U(e^i)) $ for every $i$, since $ (V_i^*, V\setminus V_i^*) $ separates the two representatives in $ f^i $ as well, and $ f(U(e^i)) $ is the weight of the minimum submodular cut between these.  Let $ \opt_{GH} $ be the connected components of $H_{k-1}$,  $ \alg=\{V_1,\dots,V_k\} $ be the partition found by Algorithm~\ref{alg:gh_greedy}, $ \alg_{GH} $ be the corresponding edges in the Gomory-Hu tree $ H $, and $ w_H $ be the weight function on $ H $. 
		
		We first argue that $f(\alg)\leq \sum_{i=1}^k d_{w_H}(V_i).$ This follows from basic properties of the Gomory-Hu tree of symmetric submodular systems, but needs some work. Let $C=\{g_1,\dots,g_{k-1}\}$ be the cut found by Algorithm~\ref{alg:gh_greedy}, which means $\sum_{i=1}^k d_{w_H}(V_i) = 2\sum_{i=1}^{k-1}w_H(g_i)$. As $f$ is symmetric, we have $2w_H(g_i) = f(U(g_i)) + f(V \setminus U(g_i))$. 
		
		Let us direct each $g_i$ away from $x$ in the tree $H$, and let $\vec{g_i}$ denote the directed edge. We then have
		\begin{align*}
			f(\alg) 
			&= 
			\sum_{i=1}^kf(V_i) 
			= 
			\sum_{i=1}^k f\left(\left(\bigcap_{\vec{g_j}\in\delta^{out}(V_i)} U(g_j)\right)\bigcap \left(\bigcap_{\vec{g_j}\in\delta^{in}(V_i)} V \setminus U(g_j)\right)\right) \\
			&\leq \sum_{i=1}^k \left(\sum_{\vec{g_j}\in\delta^{out}(V_i)} f(U(g_j)) +\sum_{\vec{g_j}\in\delta^{in}(V_i)} f(V \setminus U(g_j)) \right) =  
			\sum_{i=1}^k \sum_{g_j\in\delta(V_i)} w_H(g_j)
			=  
			\sum_{i=1}^k d_{w_H}(V_i),
		\end{align*}
		where the inequality follows by submodularity. Then
		\begin{align*}
			f(\alg)
			&\leq 
			\sum_{i=1}^k d_{w_H}(V_i)
			\leq 
			2\sum_{i=1}^{k-1} w_H(e^i)
			=  
			2\sum_{i=1}^{k-1} f(U(e^i))
			\leq 
			2\sum_{i=1}^{k-1} f(V^*_i) \\
			&\leq 
			2(1-1/k) \sum_{i=1}^{k} f(V^*_i)
			\leq 
			(2-2/k)f(\opt),
		\end{align*}
		where the second inequality holds because $\{V_1,\dots,V_k\}$ is an optimal solution of \mcmcut on $H$ by Lemma \ref{lem:tree_matroid}, and the second to last inequality holds because of the assumption that $V^*_k$ has maximum $f$ value among the classes. This concludes the proof of the theorem.
	\end{proof}
	
	%%%%%%%%%%%%%%%%
	\subsection{Intersection of Two Matroids}\label{sec:intersection}
	%%%%%%%%%%%%%%%%
	
	Lemma~\ref{lem:tree_matroid} implies the solvability on trees of the following more general problem involving two matroids: 
	
	\searchprob{\mcmcutdoub}{A graph $G=(V,E)$, edge weights $w\in \bR^E_+$, and rank-$k$ matroids $\cM_1=(V, \cB_1)$ and $\cM_2=(V, \cB_2)$ given by independence oracles, where $\cB_i$ is the family of bases of $\cM_i$.}{Find a solution to
		\begin{eqnarray*}
			&\min \left\{ \sum_{i=1}^k d_w(V_i) \colon \{ V_i\}_{i=1}^k \text{ is a partition of } V\text{ s.t. there is }  B_1\in\cB_1,B_2\in\cB_2\right.&\\ &\left.\text{ with }  |B_j\cap V_i|=1\ \text{for all}\ i\in [k], j\in  [2]\right\}.&
	\end{eqnarray*} }
	
	Here, the goal is to find a minimum multiway cut where there is a choice of terminals forming a basis of $\cM_1$, and \textit{not necessarily the same} choice of terminals forming a basis of $ \cM_2 $. 
	
	\begin{thm}\label{thm:double_tree}
		\mcmcutdoub can be solved in polynomial time if $G$ is a tree.  
	\end{thm}
	
	\begin{proof}
		We can construct $\cM_i'$ as in \eqref{eq:good} using $\cM_i$ ($i=1,2$), and we can find a minimum-weight common spanning set $C$ of $\cM_1'$ and $\cM_2'$ in polynomial time. Let $\{V_1,\dots,V_{\ell}\}$ denote the connected components of $E \setminus C$, and let $B_1 \in \cB_1$ and $B_2 \in \cB_2$ be bases such that $|B_j \cap V_i|\leq 1$ for all $i\in [\ell], j \in [2]$. If $\ell>k$, then there must be classes $V_i$ and $V_j$ such that $|(V_i \cup V_j)\cap B_1|\leq 1$ and $|(V_i \cup V_j)\cap B_2|\leq 1$. Remove $V_i$ and $V_j$ from the partition, and add $V_i \cup V_j$. Repeat this operation until $\ell=k$, and let $\{V_1',\dots, V_k'\}$ be the obtained feasible partition. Then $\sum_{i=1}^k d_w(V_i')\leq \sum_{i=1}^k d_w(V_i)$, so $\{V_1',\dots, V_k'\}$ is an optimal solution, because the boundary of any feasible partition is a common spanning set of $\cM_1'$ and $\cM_2'$.
	\end{proof}

	We get a significantly more difficult problem if the representatives for the two matroids are required to be the same, that is,  if we require the set of representatives to form a \textit{common basis} of $ \cM_1 $ and $ \cM_2 $. The problem is defined as follows.
	
	\searchprob{\mcmcutcom}{A graph $G=(V,E)$, edge weights $w\in \bR^E_+$, and rank-$k$ matroids $\cM_1=(V, \cB_1)$ and $\cM_2=(V, \cB_2)$ given by independence oracles, where $\cB_i$ is the family of bases of $\cM_i$.}{Find a solution to
		\begin{align*}
			\min \left\{ \sum_{i=1}^k d_w(V_i) \colon \{ V_i\}_{i=1}^k \text{ is a partition of } V \text{ s.t. there is } B\in \cB_1\cap \cB_2 \text{ with }  |B\cap V_i|=1 \text{ for all } i\in [k]\right\}.
		\end{align*}
	}
	
	We show that \mcmcutcom includes the problem of finding a common basis in the intersection of three partition matroids; this latter problem is known to be hard as it generalizes e.g. the 3-dimensional matching problem.
	
	\begin{thm}\label{thm:hardnessofcommon}
		The problem of verifying whether three partition matroids have a common basis reduces to the problem of deciding if the optimum of a \mcmcutcom instance is $0$.
	\end{thm}	
	\begin{proof}
		Let $ P_1,P_2,P_3 $ be partition matroids on a ground set $ S $ with partition classes $ \mathcal{U} = \{U_1,\dots, U_r\}$,  $ \mathcal{V} = \{V_1,\dots, V_r\}$, and  $ \mathcal{W} = \{W_1,\dots, W_r\} $, respectively. We encode $ P_3 $ into a tree of depth two by creating a root $ z $, a depth one vertex $ v_i $ for each $ W_i $, and leaf vertices for each element of $ s\in S $. The edges are between $ v_i $ and $ z $, and for each $ s\in W_i $, we add an edge between $ s $ and $ v_i $. We add the new vertices to the ground set of $ P_1 $ and $ P_2 $ to create $ M_1 $ and $ M_2 $: each $v_i$ is added as a loop in both matroids, while $z$ is added as a free element, i.e., it is added to every base. We assign weight 0 to the edges from the root $ z $, and all other edges get unit weights. 
		
		The only possible solution with objective value $0$ is the partition formed by $\{z\}$ and $W_i+v_i$ ($i\in [r]$). This is a feasible solution if and only if $P_1$ and $P_2$ have a common basis that intersects each $W_i$ in exactly one element, that is, $ P_1$, $P_2$, and $ P_3 $ have a common basis.
	\end{proof}
	
	%%%%%%%%%%%%%%%%
	\subsection{Paving Matroids}
	%%%%%%%%%%%%%%%%
	
	In this section, we show that \submcp 
	is equivalent to \subkpart
	when the matroid $\cM$ is a paving matroid. The key observation is that if the matroid is a rank-$k$ paving matroid, then \emph{any} $k$-partition has a valid basis representative.
	
	\begin{lem}\label{lem:paving}
		Let $\cM=(V,\cB)$ a rank-$k$ paving matroid for some $k\geq 1$ and $\cP=\{V_1,\dots, V_k\}$ be a partition of $V$. Then there exists $B\in\cB$ such that $|B\cap V_i|=1$ for $i\in[k]$.
	\end{lem}
	\begin{proof}
		Let $\cH = \{H_1,\dots,H_q\}$ be the hypergraph representation of $\cM$. Since every set of size at most $k-1$ is independent in $\cM$, there is an independent set $X=\{v_1,\dots, v_{k-1}\}$ such that $v_i\in V_i$ for $i\in[k-1]$. If $X+v$ is a basis for some $v\in V_k$, then we are done. Otherwise, for every $v\in V_k$ there exists a hyperedge $H_v\in\cH$ such that $X+v\subseteq H_v$. As $|H_u\cap H_v| =|X|= k-1$ for any distinct $u,v\in V_k$, we get that $H_u=H_v$ by the properties of the hypergraph representation of the paving matroid. Let us denote this unique hyperedge by $H$; we thus have $V_k\subseteq H$. By repeating the same argument for each partition class in place of $V_k$, we get that $H=V$, contradicting the hyperedges being proper subsets of $V$.
	\end{proof}
	
	%%%%%%%%%%%%%%%%
	\section{Greedy Splitting for \texorpdfstring{Sub-MCP}{\submcp}}\label{sec:greedysplit}
	%%%%%%%%%%%%%%%%
	
	In the analysis of the greedy splitting algorithm by Zhao, Nagamochi and Ibaraki~\cite{greedysplit}, they rely on a key lemma in their approximation guarantees for \submpart, \symsubmpart, \monsubmpart, \subkpart, \symsubkpart, and \monsubkpart.
	We show that this lemma holds even when an additional matroid constraint is imposed, and then follow the same method to prove Theorems \ref{thm:greedysplitsym}, \ref{thm:greedysplitmon}, and \ref{thm:general}; all aforementioned results of \cite{greedysplit} follow from these three theorems as well. The modified greedy splitting algorithm is described as Algorithm \ref{alg:gsa}.
	
	\begin{algorithm}[!h]
		\caption{Greedy Splitting Algorithm for \submcp}\label{alg:gsa}
		\begin{algorithmic}[1]
			\Statex \textbf{Input:} A submodular function $f\colon 2^V\to \bR_+$ and a rank-$k$ matroid $ \cM=(V,\cI) $.
			\Statex \textbf{Output:} A feasible partition $\cP$.
			\State $\cP_1 \gets \{V\}$.
			\For{$i=1,\dots, k-1$}
			\State Let $(X,W) \in \argmin \{f(X) +f(W \setminus X) - f(W)\colon X\subseteq W\in \cP_{i}, \text{there is}\ I\in \cI\text{ s.t. } I \text{ is a transversal} $ $\text{ of } (\cP_i\setminus  W) \cup \{X, W \setminus X\} \}$.
			\State $\cP_{i+1} \gets (\cP_i\setminus  W) \cup \{X, W \setminus X\} $
			\EndFor
			\State Return $\cP_k$.
		\end{algorithmic}
	\end{algorithm}
	
	\begin{remark}
		Line 3 of Algorithm \ref{alg:gsa} can be executed in polynomial time by the following subroutine.
		\begin{enumerate}
			\item For every pair $x,y$ of elements in $W$, we can decide if $x,y$ can be extended to an independent set so that we pick a single element from each class of $\cP_i\setminus W$.
			\item Now that we know which pairs $x,y$ are possible representatives, we can do the following: for each fixed pair $x,y$, find the best partition of $W$ into $X$ and $W \setminus X $ such that $x \in X$, $y \in W \setminus X$.
		\end{enumerate}
	\end{remark}
	
	\begin{lem}[Main Lemma in~\cite{greedysplit}]\label{lem:main}
		Let $\mathcal{P}_i$ be the $i^{th}$ partition returned by the greedy splitting procedure, and $\mathcal{P}= (V_1,\dots, V_i)$ be any partition satisfying the matroid constraint. Then for any submodular $f$,
		\begin{align*}
			f(\mathcal{P}_i)\leq \sum_{j=1}^{i-1} (f(V_j) + f(V \setminus V_j)) - (i-2)f(V)
		\end{align*}
	\end{lem}
	\begin{proof}
		We proceed, as in~\cite{greedysplit}, by induction on $i$. The base case $i=1$ holds trivially. In the $i^{th}$ step, we are given a partition $\mathcal{P}=(V_1,\dots, V_i)$ along with some $I_i\in \mathcal{I}$ that is also independent in the partition matroid induced by $\mathcal{P}$. Let $\mathcal{P}_{i-1}=(U_1,\dots, U_{i-1})$ be the partition created by the algorithm in the previous step,
		and $I_{i-1}\in \mathcal{I}$ be the representatives corresponding to $\mathcal{P}_{i-1}$. 
		As $|I_i|\geq |I_{i-1}|$, there must be an element $v\in I_i\setminus  I_{i-1}$ such that $I_{i-1}+v$ is independent in $\mathcal{M}$. 
		The element $v$ must be in some part of $\mathcal{P}_{i-1}$, say $U_h$, and some part of $\mathcal{P}$, say $V_j$. Let $u\in U_h$ be the element of $I_{i-1}$ in $U_h$.\\
		
		\noindent \textbf{Case 1:} Say $u\notin V_j\cap U_h $. Then $V_j\cap U_h, U_h \setminus V_j$ is a valid candidate split
		for Line 4 in Algorithm~\ref{alg:gsa}, 
		so
		\[ f(\mathcal{P}_i) - f(\mathcal{P}_{i-1}) \leq f(V_j\cap U_h) + f(U_h \setminus V_j) - f(U_h). \]
		Submodularity then implies 
		\begin{align*}
			f(\mathcal{P}_i) - f(\mathcal{P}_{i-1}) \leq f(V_j) + f(U_h \setminus V_j) - f(V_j\cup U_h)\leq f(V_j) + f(V \setminus V_j) - f(V).
		\end{align*}
		
		\noindent \textbf{Case 2:} For the other case, assume $u\in V_j\cap U_h$. Then $u \notin I_i$ since the element of $I_i$ in $V_j$ is $v$; furthermore, $I_{i-1}'=I_{i-1}+v-u$ is a valid representative for $\mathcal{P}_{i-1}$. We can repeat the same argument with $I_{i-1}'$. This process will eventually terminate, as $|I_{i-1}'\cap I_{i}|=|I_{i-1}\cap I_{i}|+1$, hence after a finite number of steps either we are in Case 1, or $I_{i-1}\subseteq I_{i}$.
		In the latter case, $u$ cannot be in $V_j$, as there is another $v\in V_j\cap I_{i}$, and we are again in Case 1.
		
		Thus, there is some $V_j\in \mathcal{P},$ such that $f(\mathcal{P}_i) - f(\mathcal{P}_{i-1}) \leq f(V_j) + f(V \setminus V_j) - f(V)$. By the induction hypothesis, when applied to $\mathcal{P}\setminus V_j$, we get
		\begin{align*}
			f(\mathcal{P}_{i-1}) \leq \sum_{j=1}^{i-2} (f(V_j) + f(V \setminus V_j)) - (i-3)f(V).
		\end{align*}
		Therefore,
		\begin{align*}
			f(\mathcal{P}_i) \leq f(\mathcal{P}_{i-1}) + f(V_j) + f(V \setminus V_j) - f(V) \leq \sum_{j=1}^{i-1} (f(V_j) + f(V \setminus V_j)) - (i-2)f(V),
		\end{align*}
		concluding the proof of the lemma.
	\end{proof}
	
	With Lemma~\ref{lem:main} in hand, the results of~\cite{greedysplit} follow immediately. In order to make the paper self-contained, we repeat them here. 
	
	\begin{proof} [Proof of Theorem~\ref{thm:general}]
		Lemma \ref{lem:main} shows that for the optimal partition $\cP=\{V_1,\dots, V_k\}$, the output $\cP_k$ of the Algorithm \ref{alg:gsa} satisfies
		\[ f(\mathcal{P}_k)\leq \sum_{j=1}^{k-1} (f(V_j) + f(V \setminus V_j)) - (k-2)f(V) \leq \sum_{j=1}^{k-1} \sum_{i=1}^k f(V_i) - (k-2)f(V) \leq (k-1)f(\cP). \]
		This concludes the proof.
	\end{proof}
	
	\begin{proof} [Proof of Theorem~\ref{thm:greedysplitsym}]
		Let $\cP^* = \{V_1^*,\dots,V_k^*\}$ be the optimal partition for \symsubmcp, ordered so that $V_1^*\leq V_2^*\leq \dots \leq V_k^*$. Then by Lemma~\ref{lem:main}, 
		\[f(\mathcal{P}_k)\leq \sum_{j=1}^{k-1} (f(V_j^*) + f(V \setminus V_j^*)) - (k-2)f(V),\]
		where $\cP_k$ is the partition returned by Algorithm \ref{alg:gsa}. Using the fact that $f$ is symmetric and submodular, we have
		\begin{align*}
			f(\mathcal{P}_k)\leq 2\sum_{j=1}^{k-1} f(V_j^*) - (k-2)f(V)\leq 2\sum_{j=1}^{k-1} f(V_j^*)\leq \left(2-\frac{2}{k}\right)f(\cP^*),
		\end{align*}
		concluding the proof of the theorem.
	\end{proof}
	
	\begin{proof} [Proof of Theorem~\ref{thm:greedysplitmon}]
		Let $\cP^* = \{V_1^*,\dots,V_k^*\}$ be the optimal partition for \monsubmcp, ordered so that $V_1^*\leq V_2^*\leq \dots \leq V_k^*$. Then, by Lemma~\ref{lem:main}, 
		\[f(\mathcal{P}_k)\leq \sum_{j=1}^{k-1} (f(V_j^*) + f(V \setminus V_j^*)) - (k-2)f(V),\]
		where $\cP_k$ is the partition returned by Algorithm \ref{alg:gsa}. Using the fact that $f$ is monotone and submodular, we have
		\begin{align*}
			f(\mathcal{P}_k)&\leq \sum_{j=1}^{k-1} f(V_j^*) + f(V\setminus V_{k-1}^*)\leq \sum_{j=1}^{k-2} f(V_j^*) + f(V_{k-1}^*) + f(V \setminus V_{k-1}^*)\\
			&\leq \left(1-\frac{2}{k}\right)f(\cP^*) + f(V_{k-1}^*) + f(V \setminus V_{k-1}^*)\leq \left(2-\frac{2}{k}\right)f(\cP^*),
		\end{align*}
		concluding the proof of the theorem.
	\end{proof}
	
	Our analysis of the $(2-2/k)$-approximation for monotone submodular functions is tight for the greedy splitting algorithm, even without the matroid constraints, by the following lemma.
	
	\begin{lem}\label{lem:tightness}
		There is an instance of monotone submodular $k$-partition for which the greedy splitting algorithm achieves exactly $(2-2/k)$-approximation.
	\end{lem}
	\begin{proof}
		Let $\cP = \{S_1,\dots ,S_k\}$ be a partition of the ground set $V$ into subsets of size at least $2$. Consider the laminar matroid $\cM$ in which a set $X\subseteq V$ is independent if and only if $|V\cap S_i|\leq 1$ for $i\in[k]$ and $|X|\leq k-1$. Let $f$ be the rank function of $\cM$ -- clearly, $f$ is monotone and submodular. 
		
		For any initial split of the greedy splitting algorithm, we have $f(S) + f(V \setminus S) \geq k$. Therefore, the algorithm might choose $S$ to be a singleton. In general, as the greedy splitting algorithm proceeds, in the $i^{th}$ step the algorithm might split the class $S_i$ into a singleton and the set of remaining elements. In such a scenario, the resulting partition has the form $\{s_1,\dots,s_{k-1},V\setminus\{s_1,\dots,s_{k-1}\}\}$. The sum of the ranks of these sets is $2k-2$, while the optimal partition is $\cP$ with a total $f$ value of $k$. This shows that greedy splitting might lead to a multiplicative error of $2-2/k$.
	\end{proof}
	
	\bibliographystyle{abbrv}
	\bibliography{arxiv}
	
\end{document}